\documentclass[onefignum,onetabnum]{siamonline250211}

\usepackage{lipsum}
\usepackage{amsfonts}
\usepackage{graphicx}
\usepackage{epstopdf}
\usepackage{algorithmic}
\usepackage{csquotes}
\usepackage{tikz}
\usepackage{cite}
\usepackage[latin1]{inputenc}
\usetikzlibrary{trees}
\usetikzlibrary{matrix,positioning,quotes}
\ifpdf
  \DeclareGraphicsExtensions{.eps,.pdf,.png,.jpg}
\else
  \DeclareGraphicsExtensions{.eps}
\fi

\definecolor{green}{cmyk}{1,0,1,0}
\definecolor{blue}{cmyk}{1,0.2,0,0}
\definecolor{red}{cmyk}{0,1,1,0}

\tikzstyle{peers}=[draw,circle,green,bottom color=white,
                  top color= white, text=black,minimum width=10pt]
\tikzstyle{superpeers}=[draw,circle,blue, left color=white,
                       text=black,minimum width=15pt]
\tikzstyle{3core}=[draw,circle,red, left color=white,
                   text=black,minimum width=10pt]                       
\tikzstyle{legendsp}=[rectangle, draw, blue, rounded corners,
                     thin,bottom color=white, top color=white,
                     text=blue, minimum width=2cm]
\tikzstyle{legendp}=[rectangle, draw, green, rounded corners, thin,
                     bottom color=white, top color= white,
                     text= green, minimum width= 2cm]
\tikzstyle{legend3core}=[rectangle, draw, red, rounded corners, thin,
                     bottom color=white, top color= white,
                     text= red, minimum width= 2cm]

\newcommand{\cc}[1]{\ensuremath{\mathbf{#1}}}
\DeclareMathOperator{\per}{perm}

\usepackage{enumitem}
\setlist[enumerate]{leftmargin=.5in}
\setlist[itemize]{leftmargin=.5in}

\newsiamremark{remark}{Remark}
\newsiamremark{hypothesis}{Hypothesis}
\crefname{hypothesis}{Hypothesis}{Hypotheses}
\newsiamthm{claim}{Claim}
\newsiamthm{observation}{Observation}

\headers{Faster network motif discovery by counting isomorphic subtrees}{J. A. Grochow, T. Tohme}

\title{Faster network motif discovery by counting isomorphic subtrees\thanks{\funding{During this work, JAG was partially funded by NSF grant DMS-1622390, and TT was funded by J. A. Grochow startup funds.}}}

\author{Tarek Tohme\thanks{American University of Beirut, and Department of Computer Science, University of Colorado Boulder
  (\email{tarektohme1@gmail.com}).}
\and Joshua A. Grochow\thanks{Departments of Computer Science and Mathematics, University of Colorado Boulder 
(\email{jgrochow@colorado.edu}).}
}

\usepackage{amsopn}

\ifpdf
\hypersetup{
  pdftitle={Faster network motif discovery by counting isomorphic subtrees},
  pdfauthor={J. Grochow, T. Tohme}
}
\fi

\begin{document}

\maketitle

\begin{abstract}
  We develop a new algorithm for counting the number of subgraphs of a network isomorphic to a given query graph (\#\textsc{SubgraphIsomorphism}), motivated by network motif search. High-degree vertices (hubs), common in real-world networks, contribute to a combinatorial explosion in the number of subgraphs, making existing motif search algorithms intractable for motif sizes greater than $\approx 8$ on a wide variety of networks of interest. Our procedure leverages the $k$-core decomposition and a novel subtree-counting technique to quickly scan the periphery of a network. These two innovations allow our algorithm to significantly speed up its predecessors in practice, especially as most real-world networks have a relatively large periphery. We prove that \#\textsc{RootedSubtreeIsomorphism}, a key subroutine in our algorithm, is \cc{\#P}-complete via a reduction from counting bipartite matchings. We provide analytic upper bounds on our algorithm's execution time, and evaluate its performance on 11 real-world networks of varying topologies.
\end{abstract}



\section{Introduction}

Networks are useful abstractions for studying structural patterns and properties in many types of systems
\cite{strogatz2001exploring, newman2018networks}. One way to study such patterns is by identifying repeated subnetworks, following the reasoning that in many systems, the frequent occurrence of a ``module" is an indicator of its importance in the function of the system. Patterns in a network that appear with statistically significant frequency with respect to a null model (an appropriately defined random network) are called \emph{network motifs} \cite{milo2002network}. Examples of networks where motifs have been studied include class diagrams in software architecture \cite{valverde2005network}, brain networks \cite{sporns2004motifs}, transcriptional regulation networks \cite{shen2002network} and many more \cite{ohnishi2010network, albert2004conserved, paulau2015motif, itzkovitz2005coarse}.

Since their introduction \cite{milo2002network}, an active line of research has been dedicated to the design of algorithms that search for such motifs ever more efficiently. Most of these algorithms fall under one of two categories: network-centric algorithms, in which all \cite{milo2002network, kashani2009kavosh, li2018mtmo, wernicke2006fanmod, AhmedPGD, santoso2020efficient} or a sample \cite{kashtan2004efficient, alon2008biomolecular, wernicke2006fanmod, RibeiroSilvaGTriesSampling} of size-$k$ subgraphs of a given network are enumerated, and motif-centric algorithms \cite{grochow2007network, mavisto, nemofinder, omidi2009moda, mbadiweKimParaMODA, RibeiroLopesSilvaParallelGTries, RibeiroSilvaGTriesSets, ribeiro-PHD2011, ribeiro2014g, ribeiro2010g} in which size-$k$ graphs are enumerated \emph{a priori} and mapped onto a network one by one. Typically in the former category, after enumeration, subgraphs are grouped into isomorphism classes and their frequencies are then calculated or estimated. Two advantages of motif-centric approaches are the possibility of searching for the instances of a single subgraph in a network, and the reduced cost of computing isomorphisms for each subgraph incurred in the classification task of network-centric algorithms. Additional ideas such as symmetry-breaking \cite{grochow2007network} and expansion trees \cite{omidi2009moda} have enabled further speedups, making it possible to search for all subgraphs up to size 8, and individual subgraphs up to size 30 \cite{grochow2007network}. The g-trie, introduced in \cite{ribeiro2010g}, is an efficient data structure that enables the search for any set of subgraphs, taking advantage of the similarity between subgraphs in that set. This minimizes the cost of redundant searches incurred by motif-centric approaches, which execute as many queries as there are subgraphs, disregarding their potential similarity, while still eliminating the unnecessary isomorphism checks of network-centric subgraph enumeration. This method used in combination with symmetry-breaking produces a vast improvement compared with pre-existing approaches \cite{ribeiro2014g}. Although an exhaustive list is beyond the scope of this paper, additional techniques have included focusing only on tree motifs \cite{li2018mtmo}, parallelizing subgraph enumeration \cite{shao2014parallel} or subgraph mapping \cite{lin2016network}.

A few remarks serve to contextualize and motivate our strategy for an improved motif search algorithm. First, we note that existing approaches to motif search all operate by either enumerating or mapping \emph{instances} of subgraphs inside a network. In principle, to discover motifs it would suffice to count all subgraphs of a network, rather than list them. The counting version of the subgraph isomorphism problem, \#\textsc{SubgraphIsomorphism}, asks how many subgraphs of a graph $G$ are isomorphic to another graph $H$. Algorithms for \#\textsc{SubgraphIsomorphism} include \cite{fomin2012faster} and \cite{alon2008biomolecular}. The latter color-coding technique achieves polynomial time when counting non-induced subgraphs of bounded treewidth, provided $|H| = O(\log(|G|))$. However it was recently proved \cite{curticapean2014complexity} that boundedness of the vertex cover number of $H$ is the only condition that guarantees tractability for \#\textsc{SubgraphIsomorphism}. Regardless of asymptotic behavior, counting subgraphs has the potential to be faster than listing them, owing to the use of clever counting techniques and to the memory access and storage overhead when large numbers of instances are involved. The ESCAPE algorithm \cite{pinar2017escape}, which manages to count all size-5 subgraphs orders of magnitude more efficiently than predecessors, is a testimony to this.

Another issue with existing algorithms is that none of them account for the difference in runtime incurred by different subgraph topologies. Early on in the investigation of network motifs, it was suggested \cite{itzkovitz2003subgraphs} that certain subgraphs could contribute disproportionately more to the runtime of motif search because of combinatorial explosions in the number of times they appear in a network. A good example is the star graph of size $k$, of which there are $\binom{n-1}{k-1}$ instances in a star graph of size $n \geq k$. Such combinatorial explosions are modeled statistically in \cite{picard2008assessing}, and further examples are discussed in \cite{grochow2007network}.

Lastly, scale-free networks are known to possess large hierarchically structured hubs \cite{barabasi2001deterministic} which can quickly make the motif search problem intractable in such networks. Putting aside the discussion of how ubiquitous scale-free networks really are \cite{broido2019scale}, biological networks, where investigations of network motifs originated, often also contain prohibitively large hubs \cite{prvzulj2004modeling} whether their degree distributions are scale-free, geometric, or log-normal. As we will see examples of experimentally, the runtime of motif search can even be dominated by a single motif while searching among hundreds or thousands of motifs, especially when the target network contains large hubs.

Taking into account the above remarks, we make \#\textsc{SubgraphIsomorphism} our focus in the present paper. We design a counting algorithm that leverages tree and hub structure in both networks and motifs, and can in principle be used as a subroutine to any standard motif-centric algorithm. In this paper, we adapt it to the Grochow--Kellis symmetry-breaking algorithm \cite{grochow2007network}. Our main contribution consists in separating both the network and the query graph into two subsets using the $k$-core decomposition: the ``periphery", which contains all the low-connectedness nodes, and the more interconnected nodes, which can be viewed as ``the core". A key feature of this decomposition is that it guarantees that the periphery consists entirely of trees. This enables us to run the peripheries of both network and query through a fast subtree-counting algorithm, saving considerable time compared with standard motif search procedures. 

In addition, the $k$-core decomposition offers a powerful early abort for motif search. The \emph{coreness} of a vertex imposes a strong constraint on the possible isomorphisms that can map a query graph to a target network, allowing our algorithm to waste less time trying out mappings that cannot yield valid instances. Complementing this result, we show that under worst-case analysis, \#\textsc{RootedSubtreeIsomorphism}---where the network $G$ and query $H$ are rooted unlabeled trees---is \cc{\#P}-complete. This suggests that one should not expect significantly improved worst-case bounds for the subtree counting portion of our algorithm, so heuristic approaches such as ours, which take advantage of structure in real-world data, may be the best one can hope for.

\paragraph{Organization} In Section \ref{sec:subtrees} we recall the definitions of core and periphery that we use, provide an algorithm for subtree-counting and analyze its time complexity. Section \ref{sec:main} introduces our main algorithm and provides an analytic treatment of its performance. We then test the algorithm on a corpus of 11 real networks, and relate its empirical performance to several key network statistics in Section \ref{sec:experiments}. In Section \ref{sec:discussion} we discuss the impact of subtree-counting and queries of high coreness on the computational complexity of motif search. We conclude in Section \ref{sec:conclusion} with a few suggestions for further extensions and applications of our method.

\paragraph{AI usage statement}
As far as the authors are aware, no LLMs or LLM-based products were used in the process of research nor writing of this paper. (We remark that almost all of this work was carried out in 2020; only final polishing of the manuscript was carried out thereafter, for which we did not use AI.)

\section{Counting subtrees in the $1$-shell}
\label{sec:subtrees}
\subsection{Background on the $k$-core decomposition}
\label{subsec:kcore}
The $k$-core decomposition \cite{kong2019k} reveals a ``core-periphery" structure to a given network. The $k$-core of a graph $G$ is the (unique) maximal subgraph $K$ such that every vertex in $K$ has at least $k$ neighbors in $K$. The \emph{$k$-shell} is the set of vertices in the $k$-core that are not part of the $k+1$-core. A vertex inside the $k$-shell is said to be of \emph{coreness $k$.} The $k$-core decomposition can be computed in linear time as follows: remove all vertices of degree $1$ until there are none left, iteratively including vertices whose degree becomes at most one during this removal process. Those removed vertices constitute the 1-shell, and the remaining vertices are the $2$-core. For each $k=2,3,\dotsc$, we continue this process: remove all vertices of degree at most $k$, including those whose degree becomes $\leq k$ during this process. The vertices removed are the $k$-shell and those left constitute the $(k+1)$-core. As we remove vertices, we can label them with their coreness for ease of future look-up.

\begin{figure}[h]
\label{fig:kcorepicture}
\centering
\begin{tikzpicture}[auto, thick]
  \foreach \place/\name in {{(0,-1)/a}, {(2,0)/b}, {(2,2)/c}, {(0,2)/d},
           {(-2,1)/e}}
    \node[superpeers] (\name) at \place {};
  \foreach \source/\dest in {a/b, a/c, a/d, b/c, c/d, a/e, d/e, b/d}
    \path (\source) edge (\dest);
  \foreach \place/\name in {{(0,-1)/f}, {(2,0)/g}, {(2,2)/h}, {(0,2)/i}}
    \node[3core] (\name) at \place {};
  \foreach \pos/\i in {above left of/1, left of/2, below left of/3}
    \node[peers, \pos = e] (e\i) {};
   \foreach \speer/\peer in {e/e1,e/e2,e/e3}
    \path (\speer) edge (\peer);
   \foreach \pos/\i in {above right of/1, right of/2, below right of/3}
    \node[peers, \pos =b ] (b\i) {};
   \foreach \speer/\peer in {b/b1,b/b2,b/b3}
   \path (\speer) edge (\peer);
   \node[peers, above of=d] (d1){};
   \path (d) edge (d1);
   \node[peers, left of=e2] (e5){};
   \path (e5) edge (e2);
   \node[peers, left of=e5] (e6){};
   \path (e6) edge (e5);
   \foreach \pos/\i in {below left of/1, below of/2}
   \node[peers, \pos =a ] (a\i) {};
   \foreach \speer/\peer in {a/a1,a/a2}
   \path (\speer) edge (\peer);
   \node[legendp] at (5,2) {\small{1-shell}};
   \node[legendsp] at (5,1) {\small{2-core}};
   \node[legend3core] at (5,0) {\small{3-core}};
\end{tikzpicture}
\caption{An example $k$-core decomposition, revealing the ``core'' and ``periphery'' of the network.}
\end{figure}
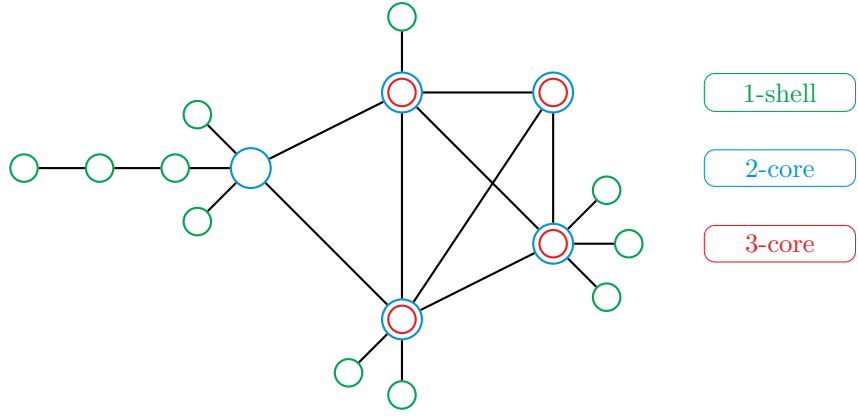

\subsection{The subtree-counting algorithm} \label{subsec:per}
The $1$-shell consists of a forest (disjoint union of trees); if a given connected component is not itself a tree, then trees in the $1$-shell of that component will have their roots belonging to the $2$-core (see Figure \ref{fig:kcorepicture}).
To see why, notice that no tree can contain a nontrivial subgraph where every vertex is of degree 2 or more. Performing a $k$-core decomposition on $H$ and $G$ reveals these peripheral trees, if they exist. By selecting one root at a time from each of $H$ and $G$'s 1-shell trees, we can then run subtree-counting as a subroutine of a motif-centric algorithm. We build on the Grochow--Kellis symmetry-breaking algorithm \cite{grochow2007network}, but in principle our subtree-counting technique can be used in combination with most standard motif-centric algorithms with little modification. In practice, there are subtleties of implementation that have to do with with symmetry-breaking and the 1-shell/2-core boundary that need to be addressed, which we discuss in Section \ref{sec:boundary}. In the remainder of this section we discuss the computational complexity of subtree-counting and motivate its use in the context of motif search. 

The permanent of an $m \times n$ matrix $A = (a_{ij})$ is defined as
\begin{equation}\label{eq:permanent}
  \per(A) = \sum_{\sigma \in \Sigma(m,n)} \prod_{i=1}^{n}a_{i,\sigma(i)}
\end{equation}
where $\Sigma(m,n)$ is the set of all injections $\{1, \dots, m\} \to \{1, \dots, n \}$. In a bipartite graph $B$ with $m+n$ vertices, we call a matching \emph{saturated} if it covers all $m$ vertices on the first side. If $B$ is a weighted bipartite graph with weighed $m \times n$ adjacency matrix $W = (w_{ij})$, and we define the weight of a saturated matching $s$ as the product of the weights of all its edges

\begin{equation}\label{eq:weighed_bipartite}
  w(s) = \prod_{i=1}w_{i,s(i)}
\end{equation}
then the permanent of $W$ is the sum of the weights of all saturated matchings of $B$.

\paragraph{A bipartite graph for counting subtrees recursively} Suppose we wish to count the number of instances of the rooted tree $(T,r)$ in the rooted tree $(S,r')$. Consider a bipartite graph where left vertices correspond bijectively to the children of $r$ in $T$, and right vertices to the children of $r'$ in $S$. For a left vertex $u$, let $T_u$ denote the subtree of $T$ rooted at $u$, and for a right vertex $v$ similarly let $S_v$ denote the subtree of $S$ rooted at $v$. In the bipartite graph, we place the edge $(u,v)$ with weight equal to the number of isomorphic copies of $T_u$ as a rooted subtree of $S_v$. Then the weighted sum of saturated matchings in this graph is the total number of rooted copies of $T$ in $S$. To determine the edge weights $(u,v)$, we can use this process recursively on the pair of rooted trees $(T_u, u)$ and $(S_v, v)$.

If the roots of trees $\{T,S\}$ have $k$ and $l$ children respectively, we need to compute $\binom{l}{k}$ permanents of $k \times k$ matrices to count all instances of $T$ in $S$. Hubs that appear in trees can make $\binom{l}{k}$ a very large number, and permanents of $k \times k$ matrices can be expensive to calculate when $k$ is large. Note that these computations happen for every non-leaf vertex of $T$, as we elaborate below. It is often the case that hubs have many degree-1 nodes connected to them---and this will be even more likely as we get to nodes further down in the tree---and we noticed experimentally that such cases were taking a large fraction of the computation time, despite being mathematically very simple. We thus take advantage of the following lemma to simplify the computation in the case of degree-1 children of the root (which can then be applied to the degree-1 children of any node during the recursive process).

\begin{observation}
Given rooted trees $(T,r)$ and $(S,r')$, let $B$ is the weighted bipartite graph above with weighted adjacency matrix $W$. If $T$ contains $m$ children at depth 1, of which $m-m'$ are leaves, then $B$ contains $m-m'$ rows that are all-1s. Furthermore, if $S$ contains $n$ children at depth $1$ of which $n-n'$ are also leaves, then up to permutations of the rows and (separately) columns, $W$ has the form
\begin{equation} \label{eq:W}
W = 
\begin{pmatrix}
	W' & \mathbf{0}  \\
	\mathbf{1} & \mathbf{1} 
\end{pmatrix}
\end{equation}
where $W'$ has shape $m' \times n'$.
\end{observation}

\begin{lemma}
\label{lem:simplified_per}
  Let $W$ be an $m \times n$ integer matrix of the form \eqref{eq:W} with parameters $m',n'$ as above.
Then
\begin{equation}\label{eq:simplified_per_eq}
    \per(W) = \per(W') \times \frac{(n-m')!}{(n-m)!}
\end{equation}
\end{lemma}

\begin{proof}
Each nonzero summand in the permanent $\per(W)$ corresponds to the following procedure: 1) first choose an injection $\iota\colon\{1,\dotsc,m'\} \to \{1,\dotsc,n'\}$---since any other choice in the last $m'$ rows would lead the summand to be zero---and compute $p_{\iota} := \prod_{i=1}^{m'} W_{i,\iota(i)}$; then 2) remove the first $m'$ rows from $W$ and the columns $\{\iota(1),\dotsc,\iota(m')\}$ from $W$, and 3) multiply $p_{\iota}$ by the permanent of what remains. Since the last $m-m'$ rows are all-1s, what remains is always an all-1s matrix of size $(m-m') \times (n-m')$. The permanent of an all-1s matrix of that size is simply the number of injections $\{1,\dotsc,m-m'\} \to \{1,\dotsc,n-m'\}$, which is $(n-m')! / (n-m)!$.
\end{proof}

Thus, If $r$ has $m$ children of which $m'$ are not leaves, and $r'$ has $n$ children, the contribution of $r$'s leaves to the total number of matchings is $\frac{(n-m')!}{(n-m)!}$. $\per(W')$ then counts the matchings of $T_u$'s rooted subtrees of depth greater than 1 in $S_v$.\\

\begin{remark} 
Note that Lemma \ref{lem:simplified_per} can be directly generalized to account for more complex special cases of the weight matrix $W$. Namely, subtrees of $T$ can be grouped into isomorphism classes, and the contribution of each class to the total number of matchings can be obtained algebraically, by reasoning over the shape and number of rooted subtrees of $S$. Implementing such counting in code seems significantly more complicated than for the case of Eq. \ref{eq:simplified_per_eq}, and we found in practice that the latter already offers large speedups when applied to real-world networks.
\end{remark}

We use the above to propose a subtree-counting algorithm that achieves drastic speedup in practice when compared with regular motif search algorithms running on trees.

\begin{algorithm}
\caption{\textsc{CountRootedSubtrees}$(T, S, t, s, D)$}
\label{alg:buildtree}
\begin{algorithmic}
\STATE{treeChildrenT := $\{i \in t$.children $\;|\; \deg(i)>1\}$}
\STATE{treeChildrenS := $\{i \in s$.children $\;|\; \deg(i)>1\}$}
\STATE{leafChildrenT := $\{i \in t$.children $\;|\; \deg(i)=1\}$}
\STATE{subtreesPreorderStrings := $\{D[i] \enspace \forall i \in $ treeChildrenT\}}
\STATE{}

\IF{$t$.children = $\emptyset$}
\RETURN 1
\ELSIF{treeChildrenT = $\emptyset$}
\STATE{$k := t$.children.size}
\STATE{$n := s$.children.size}
\RETURN $n! / (n-k)!$ 
\ENDIF
\STATE{}

\FOR{$u$ in treeChildrenT}
\FOR{$v$ in treeChildrenS}
\STATE{$T_u$ := detachTree($T, u, t)$} \quad //the ``downward-facing" subtree of $T$ rooted at $u$
\STATE{$S_v$ := detachTree($S, v, s)$}
\STATE{$W'[u][v]$ := \textsc{CountRootedSubtrees}($T_u, S_v, u, v, D$)}
\ENDFOR
\ENDFOR
\STATE{}

\STATE{equivalenceClassSizes := \{countStringsEqualTo[$k$] $\enspace \forall k \in$ subtreesPreorderStrings\}}
\STATE{multiplicity := $\prod_{n \in equivalenceClassSizes} n!$}
\STATE{}

\STATE{$n := s$.children.size $-$ treeChildrenT.size}
\STATE{$k := $leafChildrenT.size}
\RETURN $(\per(W')$ $\div$ multiplicity $) \times n! / (n-k)!$
\end{algorithmic}
\end{algorithm}

Algorithm \ref{alg:buildtree} counts all the rooted instances of an unlabeled rooted tree $T$ inside an unlabeled rooted tree $S$, where the roots of both trees are given by $t$ and $s$ respectively. $D$ is a dictionary that stores at key $u$ the preorder traversal string of the subtree of $T$ rooted at $u$ for all nodes $u \in T$ (here we abuse the term ``subtree" by having it refer to the subgraph of $T$ that contains $u$ and all its descendants relative to the root of $T$).

Among $T$'s and $S$'s depth-1 vertices, the subset of those that form a star graph along with the root (i.e. those of degree 1) are called leafChildren. The remaining vertices are the treeChildren. By Lemma \ref{lem:simplified_per}, we can compute matchings among the treeChildren of $T$ and $S$ separately, and obtain the matchings between leafChildrenT and the children of $s$ from the binomial factor.

We are interested in counting only one copy of $T$ inside $S$ for each one of $T$'s automorphism classes, since we are dealing with unlabeled trees where symmetries do not matter. To account for this, the algorithm starts by listing the preorder strings of the subtrees of $T$ rooted at $T$'s treeChildren, and grouping the identical strings among that list into equivalence classes (recall that a preorder traversal uniquely determines an unlabeled tree, up to symmetry). There are $n!$ possible permutations of the elements of an equivalence class of subtrees, whence the correction factor \emph{multiplicity}.

Next, the algorithm checks for the two base cases, where $T$ consists of only one node with no children, and where $T$ has no treeChildren. Otherwise, it loops through every pair of subtrees ($T_u$, $S_v$) of depth greater than 1 and recursively counts the instances of $T_u$ in $S_v$, storing the count in a matrix $W'$ at entry $[u][v]$. This is equivalent to constructing a directed bipartite graph and assigning a weight $k$ to edge $(u, v)$.

Finally, the permanent of $W'$ is computed using Ryser's formula \cite{ryser1963combinatorial}:
\begin{equation}
	\per(W') = \sum_{k=0}^{n-1}(-1)^k\Sigma_k
\end{equation}
where $\Sigma_k$ is obtained by summing the products of the row-sums of $A_k$, the matrix obtained by deleting $k$ rows from $W'$, for all combinations of rows. We use this expression of the permanent as a `black box', but more efficient formulas could easily be inserted in our algorithm when they become available. The worst-case time complexity of Ryser's formula is $O(2^{n-1}n^2)$ \cite{ryser1963combinatorial}. Nonetheless, to get an idea of how much faster Algorithm \ref{alg:buildtree} is than regular motif search over trees in practice, see Section \ref{sec:experiments}.

\subsection{Complexity of counting subtrees}
We refer to standard textbooks such as Arora \& Barak \cite{arora2009computational} for background on \cc{\#P} and \cc{\#P}-completeness. Here we remark that \cc{\#P}-complete problems are at least as hard, and thought to be significantly harder, than \cc{NP}-complete problems, so a \cc{\#P}-completeness result for a counting problem indicates strongly that we should not expect a subexponential-time algorithm to exist \cite{dell2014exponential}. In this section, we prove such a result for the following problem. \\
\begin{quotation}
\noindent \#\textsc{RootedSubtreeIsomorphism} \\
\textit{Input:} Two rooted trees $T$ and $S$ with their respective roots $t$ and $s$. \\
\textit{Output:} The number of rooted subtrees of $S$ that are isomorphic to $T$
\\
\end{quotation}
We note that this result is different from the one proved by Goldberg and Jerrum in 2000 \cite{goldberg2000counting}, in which counting all rooted subtrees which are \emph{distinct} up to (abstract) isomorphism within a given tree is proved \cc{\#P}-complete, whereas our result concerns counting all rooted subtrees of $S$ isomorphic to a given rooted tree $T$. \\

For $(S,s)$ a rooted tree with root $s$, a rooted subtree is a tree $(T,t)$ with the same root $t=s$ and such that $V(T) \subseteq V(S)$. For rooted trees $T$ and $S$, we write $T \preceq S$ if there exists a root-preserving isomorphism between tree $T$ and a subtree of $S$, and $T \preceq ! \, S$ if such an isomorphism is unique up to the symmetries of $T$, \textit{i.\,e.}, there is a unique rooted subtree of S that is rooted-isomorphic to $T$.

We recall here the definition of \cc{\#P} (e.\,g., \cite{arora2009computational})  for convenience:
\begin{definition} \label{def:sharpp}
	A function $f : \{0,1\}^* \mapsto \mathbb{N}$ is in \cc{\#P} if there exists a polynomial-time Turing machine $M$ such that for all $x \in \{0,1\}^*$, $f(x)$ is the number of polynomial-size certificates $y$ satisfying $M(x, y) = 1$.
\end{definition}

\begin{theorem}\label{thm:subtree_count}
\#\textsc{RootedSubtreeIsomorphism} is \cc{\#P}-complete.
\end{theorem}

\begin{proof}
That \#\textsc{RootedSubtreeIsomorphism} is in \cc{\#P} is quite direct. The polynomial-time machine $M(x,y)$ as in the definition above works as follows. Following the notation of Definition~\ref{def:sharpp}, $x$ consists of the pair of rooted trees $(T,S)$, and the machine $M$ treats $y$ as a subset of $V(S)$. $M$ verifies that $s\in y$, and that the induced subgraph $S[y]$ is isomorphic to $T$. If so, $M$ accepts, and if not, $M$ rejects. 

\begin{figure}[tphb]\label{fig:equiv_example}
	\centering
	\tikzstyle{level 1}=[level distance=2cm, sibling distance=3.2cm]
	\tikzstyle{level 2}=[level distance=2cm, sibling distance=1cm]
	\tikzstyle{level 3}=[level distance=2cm, sibling distance=0.4cm]
	\tikzstyle{level 4}=[level distance=1cm, sibling distance=0.4cm]
	\tikzstyle{level 5}=[level distance=1cm, sibling distance=0.4cm]
	
	\tikzstyle{end} = [circle, minimum width=3pt,fill, inner sep=0pt]
	\begin{center}
	\begin{tikzpicture}[grow=right, scale=0.5, baseline=0cm]
    \node[end, label=left:{$T$}] {}
        child {
            node[end, label=below:{$T_{u_3}$}, color=blue] {}        
            child [color=blue]{
                    node[end] {}
                    child {
                    	node[end, label=right:{}] {}
                    	child [color=red]{
                    	    node[end] {}
                    	    child {
                    	        node[end] {}
                    	    }
                    	}
                    	edge from parent
                	}
                    edge from parent
                }
                child [color=blue]{
                    node[end, label=right:{}] {}
                    child {
                    	node[end, label=right:{}] {}
                    	child [color=red]{
                    	    node[end] {}
                    	    child {
                    	        node[end] {}
                    	    }
                    	}
                    	edge from parent
                	}
                    edge from parent
                }
                child [color=blue]{
                    node[end] {}
                    child {
                    	node[end, label=right:{}] {}
                    	child [color=red]{
                    	    node[end] {}
                    	    child {
                    	        node[end] {}
                    	    }
                    	}
                    	edge from parent
                	}
                    edge from parent
                }
            edge from parent    
        }
        child {
            node[end, label=above:{$T_{u_2}$}, color=blue] {}        
            child [color=blue]{
                    node[end, label=below:{}] {}
                    child {
                    	node[end, label=right:{}] {}
                    	child [color=red]{
                    	    node[end] {}
                    	}
                    	edge from parent
                	}
					child {
                    	node[end, label=right:{}] {}
                    	child [color=red]{
                    	    node[end] {}
                    	}
                    	edge from parent
                	}
                    edge from parent
                }
                child [color=blue]{
                    node[end, label=above:{}] {}
                    child {
                    	node[end, label=right:{}] {}
                    	child [color=red]{
                    	    node[end] {}
                    	}
                    	edge from parent
                	}
					child [color=blue]{
                    	node[end, label=right:{}] {}
                    	child [color=red]{
                    	    node[end] {}
                    	}
                    	edge from parent
                	}
                    edge from parent
                }
            edge from parent      
        }
        child {
            node[end, label=above:{$T_{u_1}$}, color=blue] {}        
            child [color=blue]{
                    node[end, label=right:{}] {}
                    child {
                    	node[end, label=right:{}] {}
                    	edge from parent
                	}
					child {
                    	node[end, label=right:{}] {}
                    	edge from parent
                	}
					child {
                    	node[end, label=right:{}] {}
                    	edge from parent
                	}
                    edge from parent
                }
            edge from parent    
        };
    \end{tikzpicture}
    \begin{tikzpicture}[thin,amat/.style={matrix of nodes,nodes in empty cells,
      row sep=2em,
      nodes={draw,solid,circle,execute at begin node={$u_{\the\pgfmatrixcurrentrow}$}}},
      amat2/.style={matrix of nodes,nodes in empty cells,
      row sep=2em,
      nodes={draw,solid,circle,execute at begin node={$v_{\the\pgfmatrixcurrentrow}$}}},
      fsnode/.style={},
      ssnode/.style={}, baseline=0cm]
    
     \matrix[amat,nodes=fsnode,label=above:$L$] (mat1) {\\
     \\
     \\};
    
     \matrix[amat2,right=2cm of mat1,nodes=ssnode,label=above:$R$] (mat2) {\\
     \\ 
     \\};
    
     \draw  (mat1-1-1) edge (mat2-1-1)
     		(mat1-1-1) edge (mat2-2-1)
            (mat1-1-1) edge (mat2-3-1)
	        (mat1-2-1) edge (mat2-1-1)
	        (mat1-3-1) edge (mat2-1-1)
	        (mat1-3-1) edge (mat2-2-1);
	 
    \end{tikzpicture}
    \,\,
    \begin{tikzpicture}[grow=left, scale=0.5, baseline=0cm]
    \node[end, label=right:{$S$}] {}
        child {
            node[end, label=above:{$S_{v_1}$}, color=blue] {}        
            child [color=blue]{
                    node[end] {}
                    child {
                    	node[end, label=right:{}] {}
                    	edge from parent
                	}
					child {
                    	node[end, label=right:{}] {}
                    	child [color=red]{
                    	    node[end] {}
                    	}
                    	edge from parent
                	}
					child {
                    	node[end, label=right:{}] {}
                    	child [color=red]{
                    	    node[end] {}
                    	    child {
                    	        node[end] {}
                    	    }
                    	}
                    	edge from parent
                	}
                    edge from parent
                }
                child [color=blue]{
                    node[end, label=right:{}] {}
                    child {
                    	node[end, label=right:{}] {}
                    	child [color=red]{
                    	    node[end] {}
                    	}
                    	edge from parent
                	}
					child {
                    	node[end, label=right:{}] {}
                    	child [color=red]{
                    	    node[end] {}
                    	    child {
                    	    node[end] {}
                    	}
                    	}
                    	edge from parent
                	}
                    edge from parent
                }
                child [color=blue]{
                    node[end] {}
                    child {
                    	node[end, label=right:{}] {}
                    	child [color=red]{
                    	    node[end] {}
                    	    child {
                    	    node[end] {}
                    	}
                    	}
                    	edge from parent
                	}
                    edge from parent
                }
            edge from parent         
        }
        child {
            node[end, label=above:{$S_{v_2}$}, color=blue] {}        
            child [color=blue]{
                    node[end, label=right:{}] {}
                    child {
                    	node[end, label=right:{}] {}
                    	edge from parent
                	}
					child {
                    	node[end, label=right:{}] {}
                    	edge from parent
                	}
					child {
                    	node[end, label=right:{}] {}
                    	child [color=red]{
                    	    node[end] {}
                    	    child {
                    	    node[end] {}
                    	}
                    	}
                    	edge from parent
                	}
                    edge from parent
                }
                child [color=blue]{
                    node[end, label=right:{}] {}
                    child {
                    	node[end, label=right:{}] {}
                    	child [color=red]{
                    	    node[end] {}
                    	    child {
                    	    node[end] {}
                    	}
                    	}
                    	edge from parent
                	}
                    edge from parent
                }
                child [color=blue]{
                    node[end, label=right:{}] {}
                    child {
                    	node[end, label=right:{}] {}
                    	child [color=red]{
                    	    node[end] {}
                    	    child {
                    	    node[end] {}
                    	}
                    	}
                    	edge from parent
                	}
                    edge from parent
                }
            edge from parent    
        }
        child {
            node[end, label=below:{$S_{v_3}$}, color=blue] {}        
            child [color=blue]{
                    node[end, label=right:{}] {}
                    child {
                    	node[end, label=right:{}] {}
                    	edge from parent
                	}
					child {
                    	node[end, label=right:{}] {}
                    	edge from parent
                	}
					child {
                    	node[end, label=right:{}] {}
                    	edge from parent
                	}
                    edge from parent
                }
            child [color=blue]{
                    node[end, label=right:{}] {}
                    edge from parent
                }
            child [color=blue]{
                    node[end, label=right:{}] {}
                    edge from parent
                }
            edge from parent     
        };
    \end{tikzpicture}
    \end{center}
  \caption{Graphical example of property \eqref{eq:bipartite_equiv}.}
\end{figure}
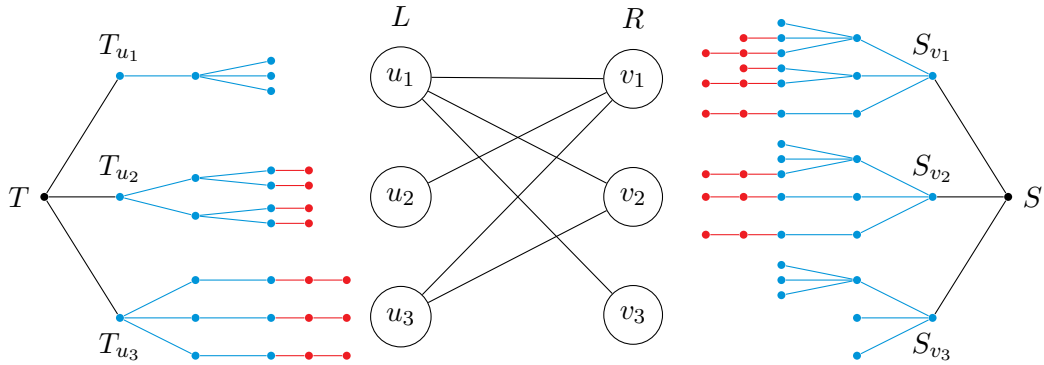

To show \cc{\#P}-hardness, we proceed via a reduction from the known \cc{\#P}-complete problem of \#\textsc{PerfectMatchings} \cite{valiant1979complexity}. A \emph{perfect matching} in a graph with $2n$ vertices is a set of exactly $n$ pairwise disjoint edges. \\
\begin{quote}
\quad \quad \noindent \#\textsc{PerfectMatchings}

\quad \quad \textit{Input:} A bipartite graph $B$ with $2n$ vertices

\quad \quad \textit{Output:} The number of perfect matchings in $B$ \\
\end{quote}

Let $B = (L, R, E)$ be a bipartite graph where $L$ and $R$ are, respectively, the sets of left and right vertices, and $E \subseteq L \times R$ is the set of edges. The main idea is to construct a tree $T_u$ for each vertex $u \in L$, and a tree $S_v$ for each $v \in R$ (the colored portions in Figure \ref{fig:equiv_example}) such that for all $u \in L, v \in R$, 
\begin{equation}\label{eq:bipartite_equiv}
  (u,v) \in E \implies T_u \preceq! \, S_v \;\; \text{ and }  \;\; (u,v) \notin E \implies T_u \not \preceq \, S_v 
\end{equation}
If these trees are joined together by their roots as in Figure \ref{fig:equiv_example} to form two trees $T$ and $S$, then the number of rooted subtrees of $S$ isomorphic to $T$ is equal to the number of matchings of $B$.

Here we show how to construct trees that have property  \eqref{eq:bipartite_equiv}, starting with the portion colored blue in Figure \ref{fig:equiv_example}. For every vertex of $L$ indexed by $i$, $T_{u_i}$ is constructed to have $i$ depth-1 vertices, and each of those has $|L| - i + 1$ children (we refer to roots of the trees $T_{u_i}$ as `depth-0' vertices, and their children as `depth-1' vertices). For the $\{S_v\}$ trees, for every vertex in $R$ indexed by $j$, $S_{v_j}$ has exactly $|L|$ depth-1 vertices. We then populate the children of those vertices by ``inserting'' one copy of $T_{u_i}$ into $S_{v_j}$ if and only if $(u_i, v_j) \in E$; in other words, $i$ among $S_{v_j}$'s depth-1 vertices has at least $|L| - i + 1$ children, and the remaining $|L| - i$ depth-1 vertices have no children. This construction guarantees that at least one instance of $T_{u_i}$ appears in $S_{v_j}$ whenever $(u_i, v_j) \in E$.

To eliminate the possibility of multiple instances of $T_{u_i}$ appearing in $S_{v_j}$ for any $i$ and $j$, such as, for instance, between $T_{u_3}$ and $S_{v_1}$ in Figure~\ref{fig:equiv_example}, we append chains of vertices to their leaves (colored red) which constrain the number of valid isomorphisms between the two trees to be at most one. For each $i$, attach a chain of length $i-1$ to each leaf of $T_{u_i}$. For every depth-2 vertex $w$ of tree $S$, if there is an edge $(u_i, v_j)$ in $B$, attach a chain of length $i-1$ to the $i$-th child of $w$, if it has any. \\

This establishes a bijection between the matchings of $B$ and the rooted subtrees of $S$ isomorphic to $T$. The bijection is given as follows: \\
\begin{itemize}
\item Given a matching $\psi : L \to R$ of $B$, for each $u \in L$, $T_{u}$ is isomorphic to exactly one subtree of $S_{\psi(u)}$, and the union of each of these subtrees along with the root $s$ constitutes the unique instance of $T$ within $S$ corresponding to $\psi$.

\item Conversely, given a rooted subtree $S' \subseteq S$ isomorphic to $T$, the matching consists of edges $(u_i,v_i) \in L \times R, i \in \{1,\dots,|L|\}$, each corresponding to a map between the subtree $S'_{v_i} \subseteq S'$ and $T_{u_i}$, the existence of which is guaranteed by the construction. \\
\end{itemize}
  
Valiant has shown that counting the number of perfect matchings is \cc{\#P}-complete \cite{valiant1979permanent}. Now if a Turing machine can query an oracle for \#\textsc{RootedSubtreeIsomorphism} on instance $(T, S)$, the answer is equal to the number of bipartite matchings in graph $B$, thus establishing the result.
\end{proof}

Like \textsc{BipartiteMatching}, this result is surprising considering that its corresponding decision problem, \textsc{RootedSubtreeIsomorphism}, is in \cc{P}.

\section{The subgraph-counting algorithm}
\label{sec:main}

\subsection{Boundary conditions}
\label{sec:boundary}
Recall from the introduction that our method consists in separating the 1-shell from the 2-core and performing subtree-counting on the former, then running a regular motif search algorithm on the latter. In doing so we must take care of a few boundary cases that arise at the interface between the 1-shell and the 2-core. To see why, let $H$ be a query graph and $G$ a network, and let $C_H$ be the 2-core (``C'' for ``core'') and $P_H$ be the 1-shell (``P'' for ``periphery'') of $H$; similarly define $C_G,P_G$.
An instance of $H$ could appear split between the 1-shell and the 2-core of $G$ such that only part of $P_H$ extends into $P_G$, but part of $P_H$ is in $C_G$.  Because of this, to properly count every instance of $H$ in $G$, we must run each tree $T_H \in P_H$ as well as all of its subtrees through Algorithm~\ref{alg:buildtree}. By using the symmetry-breaking technique introduced in \cite{grochow2007network} on $T_H$, we can cut down the cost incurred by enumerating all possible subtrees by an amount proportional to the product of the number of symmetries at each level of trees of $P_H$.

On the flip side, another boundary condition provides us with a powerful early abort when searching for instances of $H$ inside $G$: suppose that $H$ appears split between $C_G$ and $P_G$, only this time $C_H$ extends partially into $P_G$. As we prove below, this cannot happen. In fact, it cannot happen between a $k$-core in $H$ and an $h$-shell in $G$ for any integers $h < k$ without violating induced isomorphism, which we require in the formal definition of ``an instance of $H$ inside $G$''.

\begin{definition} \label{def:induced_iso}We say that an instance of $H = (V_H, E_H)$ appears in $G=(V_G, E_G)$ when there exists an injection $\phi:V_H \mapsto V_G$ between the nodes of $H$ and those of $G$ that is an isomorphism onto the subgraph induced by its image: 
\begin{equation}
\forall u,v \in V_H, \enspace (u,v) \in E_H \iff (\phi(u), \phi(v)) \in E_G
\end{equation}
\end{definition}

\begin{proposition}\label{thm:early_abort}
   If $\phi\colon H \to G$ is an isomorphism onto an induced subgraph, then 
   \[
   coreness_H(v) \leq coreness_G(\phi(v)) \qquad (\forall v \in V(H))
   \]
\end{proposition}

\begin{proof}
Let $K$ be the $k$-core of $H$. Then, by the definition of $k$-core, every vertex of $K$ has at least $k$ neighbors within $K$. Since $\phi$ is an isomorphism onto its image, we have that every vertex of $\phi(K)$ contains at least $k$ neighbors within $\phi(K)$, hence $\phi(K)$ is contained in the $k$-core of $G$.
\end{proof}

We can use this fact to create an early abort by coreness: our algorithm should never bother to try to map a vertex in $H$ to another in $G$ with lesser coreness. 

\begin{remark}
One may wonder whether the onion decomposition \cite{hebert2016multi}, a more fine-grained generalization of the $k$-core decomposition, could provide an even more powerful early abort. This may indeed be possible, but must be done with some care (the ``global'' onion layer must be used, rather than its layer within its shell, since the latter is not necessarily increased by an isomorphism onto an induced subgraph). As early-abort using coreness was already quite effective, and we were computing at least the $2$-core anyway for our subtree-counting approach, we did not pursue this direction, though it might be interesting fodder for future work.
\end{remark}

\subsection{Main algorithm}
We build on the Grochow--Kellis symmetry-breaking algorithm, and refer to the functions from their paper by the same name. Algorithm \ref{alg:countsubgraphs} takes in two graphs $H$ and $G$, and returns the number of images of $H$ that appear in $G$. It starts by running a $k$-core decomposition on both graphs. Then, as in the Grochow--Kellis algorithm, it finds all automorphisms from $H$ onto itself by calling \textsc{FindSubgraphInstances($H, H$)} \cite{grochow2007network}, groups the nodes that map to each other in the same equivalence class and returns a representative node from each class. We abbreviate this behavior in the function \emph{findEquivalenceClassRepresentatives}. The symmetry-breaking conditions are also computed similarly as in \cite{grochow2007network}.

The next part of the algorithm consists in listing all the subtrees of the 1-shell of $H$, up to symmetry. This is done by the function \emph{enumerateSubtrees}, which stores them in the dictionary \emph{subtrees}, keyed by the nodes in $H$ at which they are rooted. It does this by again making use of symmetry-breaking, and listing only one tree from each equivalence class of symmetries. One could also use standard \textsc{Rooted Tree Isomorphism} algorithms \cite{valiente2002tree}, but for simplicity of the code we use the same isomorphism algorithm for this step as well. Since $H$ is usually small, in practice we find that this step is never the bottleneck.

\begin{algorithm}
\caption{\textsc{CountSubgraphs}$(H, G)$}
\label{alg:countsubgraphs}
\begin{algorithmic}
\STATE{//Initialization}
\STATE{kCoreDecompose(H)}
\STATE{kCoreDecompose(G)}
\STATE{$ECR_H$ := findEquivalenceClassesRepresentatives($H$)}
\STATE{$SBC_H$ := findSymmetryBreakingConditions($H$)}
\STATE{$f$ := initializeParitalMap}
\STATE{}
\STATE{//enumerating all subtrees in $H$'s 1-shell}
\FOR{$h$ in $ECR_H$}
	\IF{$h \in P_H$ OR $h$ is a root} 
		\STATE{$T$ := detachTree($H, h$[, parent($h$)])}
		\STATE{subtrees[$h$] := enumerateSubtrees($T$, $h$)}
	\ENDIF
\ENDFOR
\STATE{}
\STATE{//counting subtrees in $G$'s 1-shell and boundary}
\FOR{$g$ in kCoreTraversalOrder($G$)}
	\IF{$g \in P_G$ AND $H$ is a tree}
		\STATE{$S$ := detachTree($G$, $g$, parent($g$))}
		\FOR{$h$ in $ECR_H$}
			\STATE{1ShellCount := \textsc{CountRootedSubtrees}($H, h, S, g$)}
		\ENDFOR
	\STATE{}
	\ELSIF{$g$ is a root}
		\STATE{$S$ := detachTree($G, g$)}
		\FOR{$h$ in $ECR_H$}
			\FOR{$T_h$ in subtrees[$h$]}
				\STATE{boundaryMultiplier := \textsc{CountRootedSubtrees}($T_h, h, S, g$)}
				\STATE{$f$.setMultiplier(boundaryMultiplier)}
				\STATE{$G$.markNodes($S$)}
				\STATE{\textsc{IsomorphicExtensions($f, H \backslash T_h, G, SBC_H$)}}
			\ENDFOR
		\ENDFOR
		\STATE{$G$.removeNodes($S$)}
	\ENDIF
\ENDFOR
\STATE{}
\STATE{//counting subgraphs in the 2-core}
\STATE{2CoreCount := count(\textsc{FindSubgraphInstances($H, G$)})}
\STATE{boundaryCount := $\prod_{f \in boundary} f$.getMultiplier}
\STATE{}
\RETURN 1ShellCount + boundaryCount + 2CoreCount

\end{algorithmic}
\end{algorithm}
The algorithm proceeds to count the subtrees of $H$ that appear in $G$'s 1-shell and at the boundary with its 2-core. It loops through all of $G$'s nodes in the same order as the $k$-core decomposition to quickly scan the 1-shell nodes and then remove them, making the subsequent 2-core search more tractable. We make use of the function $f$, which represents the partial isomorphism between a subset of $H$ and $G$ as the algorithm progresses in its search. For the boundary case, we associate each $f$ with a multiplier, which simply consists of a count of the number of subtrees that ``extend" from $f$ into the 1-shell of $G$. \textsc{IsomorphicExtensions} is a function introduced in \cite{grochow2007network} which incrementally extends $f$ one vertex at a time until it maps the entirety of $H$ onto a subset of $G$. By using the multiplier, we can initialize $f$ from a vertex at the boundary between the 1-shell and the 2-core of $G$, and let it expand inward into the 2-core, and use \textsc{CountRootedSubtrees} to count partial instances of $H$ in the 1-shell. Such boundary vertices are useful in the exposition, so we refer to them with a new term:

\begin{definition}
A \emph{1-shell root} of a graph $G$ is a vertex $v \in G$ in the 2-core that has at least one neighbor in the 1-shell.
\end{definition}
The 1-shell roots are precisely the roots of the trees appearing in $G$'s 1-shell.

At the end of the search, we multiply the number of instances found by \textsc{IsomorphicExtensions} by their individual multipliers to recover the total number of boundary instances. The vertices scanned by the subtree-counting routine are marked to prevent $f$ extending into the 1-shell. After all the 1-shell vertices are taken care of, they are removed and only the 2-core of $G$ remains.

Note that here we use a slightly modified version of \textsc{IsomorphicExtensions} that accounts for the early abort by coreness feature discussed above. Namely, when searching for an image vertex $g \in G$, \textsc{IsomorphicExtensions} first checks the coreness of $g$ to ensure that it is greater or equal than the coreness of its tentative preimage $h$.

\subsection{Upper bounds on runtime}
The subtree-counting algorithm works by recursively computing the number of matchings at each level of a query tree $T$ and a target tree $S$. To obtain an analytic upper bound on the number of steps in its execution, we count the total number of calls to \textsc{CountRootedSubtrees}, and multiply that by the worst-case runtime of computing the permanent of a weighed adjacency matrix (see Subsection \ref{subsec:per}). Note that in general such a matrix will not be square, which leads to an added binomial coefficient in the permanent's runtime contribution. This is because the permanent of a non-square matrix is the sum of the permanents of each of its square submatrices.

The following is an approximate expression of the number of steps in \textsc{CountRootedSubtrees}' execution on trees $T$ and $S$ . $T'$ is the set of non-leaf vertices of $T$, $\phi$ is the root-preserving isomorphism mapping $T$ to a subtree of $S$, $m_u$ and $n_v$ are the number of children of vertices $u \in T$ and $v \in S$ respectively:

\begin{equation}
\label{eq:subtrees_upperbound}
f_{crs}(T, S) \approx \sum_{u \in T'} \binom{n_{\phi(u)}}{m_u} \times m_u^2 \times 2^{m_u - 1}
\end{equation}

By comparison, the Grochow--Kellis algorithm will test all $m$ vertices of $T$ with all $n$ vertices of $S$ up to symmetries of $T$.

\begin{equation}
\label{eq:gk_upperbound}
f_{gk}(T, S) \approx \frac{1}{|Aut(T)|}\binom{n}{m} = \frac{1}{|Aut(T)|}\binom{1+\sum_{u \in T'} n_{\phi(u)}}{1+\sum_{u \in T'} m_u}
\end{equation}

Note that the right hand side of Eq. \ref{eq:subtrees_upperbound} is a sum over binomial coefficients of $n_v$ and $m_u$, whereas in Eq. \ref{eq:gk_upperbound} it is a binomial coefficient of the sum, which in general is exponentially worse. 

In the context of counting subgraphs, \textsc{CountRootedSubtrees} will be called once for each pair of non-leaf vertices in $H$ and 1-shell roots in $G$. To factor in the speedup afforded by the coreness condition described in Theorem \ref{thm:early_abort}, let $m = |H|$ and $n = |G|$, let $m_K$, $n_K$ be the sizes of the $k$-cores of $H$ and $G$, and $R$ the number of 1-shell roots in $G$. This gives an upper bound on the total runtime of the subgraph counting algorithm:

\begin{equation}
\label{eq:coreness_upperbound}
f_{cs}(H, G) \leq \frac{R}{|Aut(H)|} \times \sum_{T \subseteq P_H} f_{crs}(T, S) \times \binom{n_K}{m_K} \binom{n_{K-1} - m_K}{m_{K-1} - m_K} \dots \binom{n_2 - m_3}{m_2 - m_3}
\end{equation}

We compare this with an upper bound on the Grochow--Kellis algorithm's performance:

\begin{equation}
f_{gk}(H, G) \leq \frac{1}{|Aut(H)|}\binom{n}{m}
\end{equation}

Again, a binomial coefficient of a sum is in general much larger than the product of the binomial coefficients of the summands. \\

\section{Real-data experiments}
\label{sec:experiments}

\subsection{Methods}
We implemented the Grochow--Kellis algorithm and Algorithm \ref{alg:countsubgraphs} in Python 3.6 and ran them on a set of Intel Xeon E5-2680 v3 processors at 2.50GHz running Red Hat Enterprise Linux 7, from CU Boulder's RMACC supercomputer. Code is available at \cite{code}. Because each motif/network query is a computation that can be run on a separate processor, we made use of concurrency wherever possible. The runtimes we report in the following section are expressions of the total time taken for all concurrent processes summed together. We note that although Grochow--Kellis is no longer the state of the art \cite{patra2020review}, the subtree-counting algorithm can easily be incorporated into other algorithms as well. The point of the experiments here is to demonstrate, in as simple a setting as possible, the potential gains to be had from employing this strategy. We expect similar gains in adding the subtree counting strategy onto other motif-based counting algorithms, such as \cite{ribeiro2014g}.

Using data from the communityFitNet corpus \cite{ghasemian2017evaluating} we tested a sample of 8 undirected networks of comparable sizes (around 150 vertices) spanning different topologies and categories of systems. We took care in our selection (by visual inspection) to make sure each network's topology was typical of the networks that represent the same type of system in the communityFitNet dataset. In each network we searched for all undirected graphs up to size 8 (Figure \ref{fig:runtime_comparison}). Informed by the upper bounds derived in the previous section, we then compared the speedup afforded by our algorithm to the following network statistics: size of the 1-shell, maximum degree and average non-leaf degree in the 1-shell, maximum and average coreness of 2-core, and maximum and average tree depth in 1-shell. We identified the first three statistics as being most predictive of a large performance improvement in our algorithm. We computed these statistcs for all 442 undirected networks from the communityFitNet corpus, highlighting those with large, dense 1-shells and peripheral hubs as instances where our algorithm is likely to yield very large speed-ups. Finally, we tested three networks of around 1000 vertices, and demonstrate our algorithm's performance in querying all undirected graphs up to size 6.

\subsection{Results and discussion}
\label{sec:discussion}
Figure \ref{fig:runtime_comparison} shows our selection of 8 undirected networks along with network statistics we considered as potential indicators for our algorithm's improvement. Figure \ref{fig:lsample_comparison} shows runtimes on three larger networks for additional benchmarking. At first sight, the contribution of subtree-counting to our algorithm's speedup is more evident than that of early abort by coreness. We found that the speedup gained from queries of high maximal coreness is negligible compared with that of queries of high maximal hub degree (or star graphs). A hint to this lies in the fact that products of binomial coefficients of terms of a sum, such as in Eq. \ref{eq:coreness_upperbound}, are vastly larger than sums of binomial coefficients of terms of a sum, as in Eq. \ref{eq:subtrees_upperbound}, but both are much smaller than binomial coefficients involving the entire sum (Eq. \ref{eq:gk_upperbound}).

Our aim in these experiments was to provide a common-platform comparison with a standard motif-centric algorithm from the literature. Because our code is written in Python, it was expected that our ``wall-clock'' runtimes are not comparable to those of available motif-finding software, which is often written in faster languages such as Java \cite{grochow2007network} or C++ \cite{ribeiro2014g} and optimized for experimental use. Furthermore, we strongly suspect that the super-exponential runtime on large hubs such as the \emph{C. elegans} interactome (Figure \ref{fig:runtime_comparison}) is due to the combinatorial number of subgraphs being counted, and not to the difficulty in counting them. To illustrate this, we note that the most time-consuming size-8 subgraph to count on the \emph{C. elegans} interactome network (of which there are 11117) took only 214 seconds for our algorithm. By comparison, the most expensive of all 21 size-5 subgraphs took 2459 seconds for Grochow--Kellis. We expect that an implementation in a faster language would yield similar relative speed-ups compared to previous algorithms in the same language, and thus even larger overall speed-ups.

Our algorithm and subtree-counting technique was designed for undirected graphs, but can easily be extended to include directed and labeled graphs. 

\begin{figure}[H]
\label{fig:runtime_comparison}
	\centering
	\includegraphics[width=15.6cm]{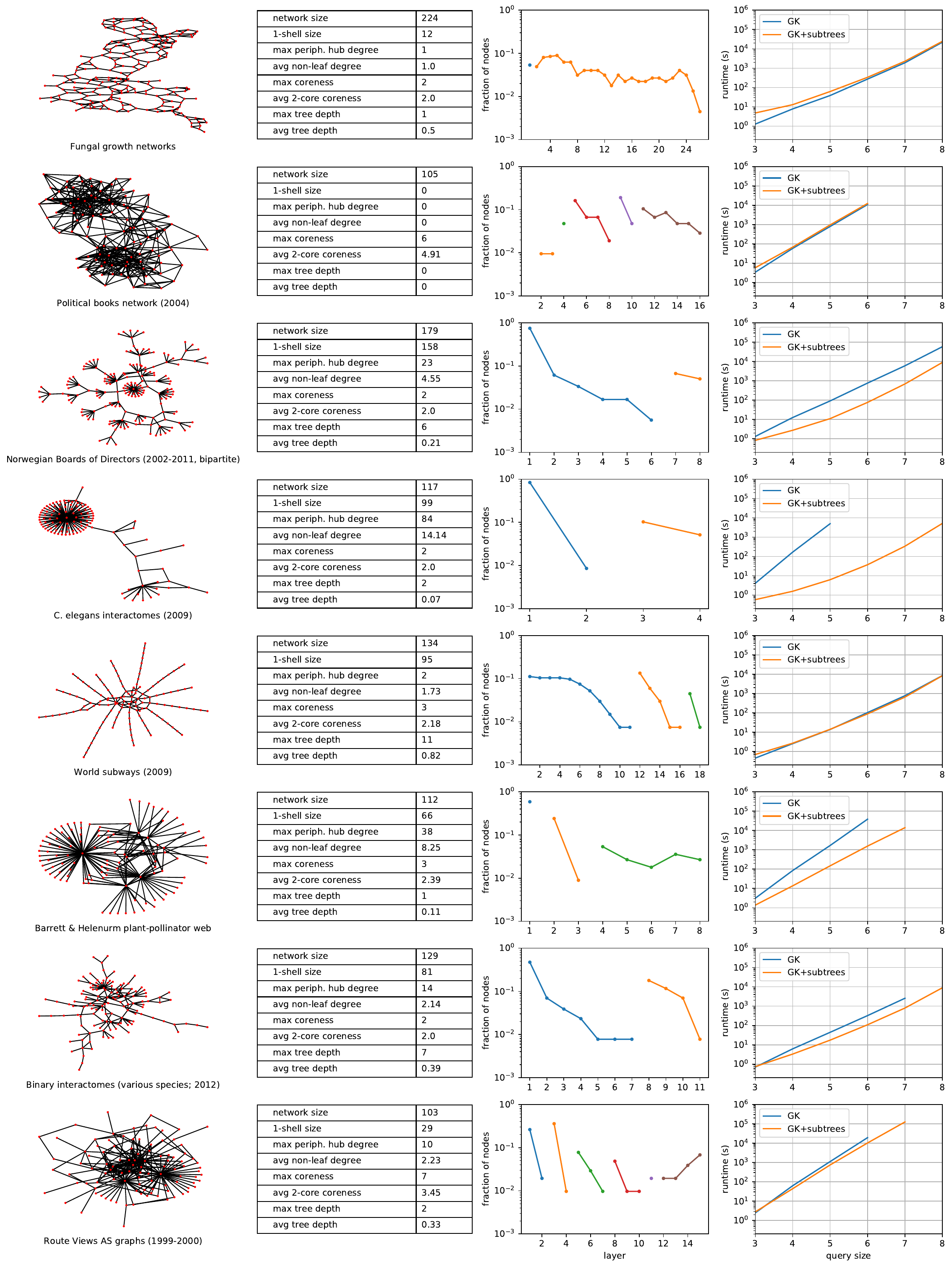}
	\caption{The first column shows each network from our sample. In the second column, \emph{max periph. hub degree} is the maximum degree of any vertex in the trees of the 1-shell, including their roots. The ``degree" of a 1-shell root is taken as the number of adjacent vertices in the 1-shell. The third column shows the onion spectrum \cite{hebert2016multi} of each network, and the last column shows a runtime comparison between our algorithm and the Grochow--Kellis symmetry-breaking algorithm.}
\end{figure}

\begin{figure}[h]
\label{fig:lsample_comparison}
	\centering
	\includegraphics[width=10cm]{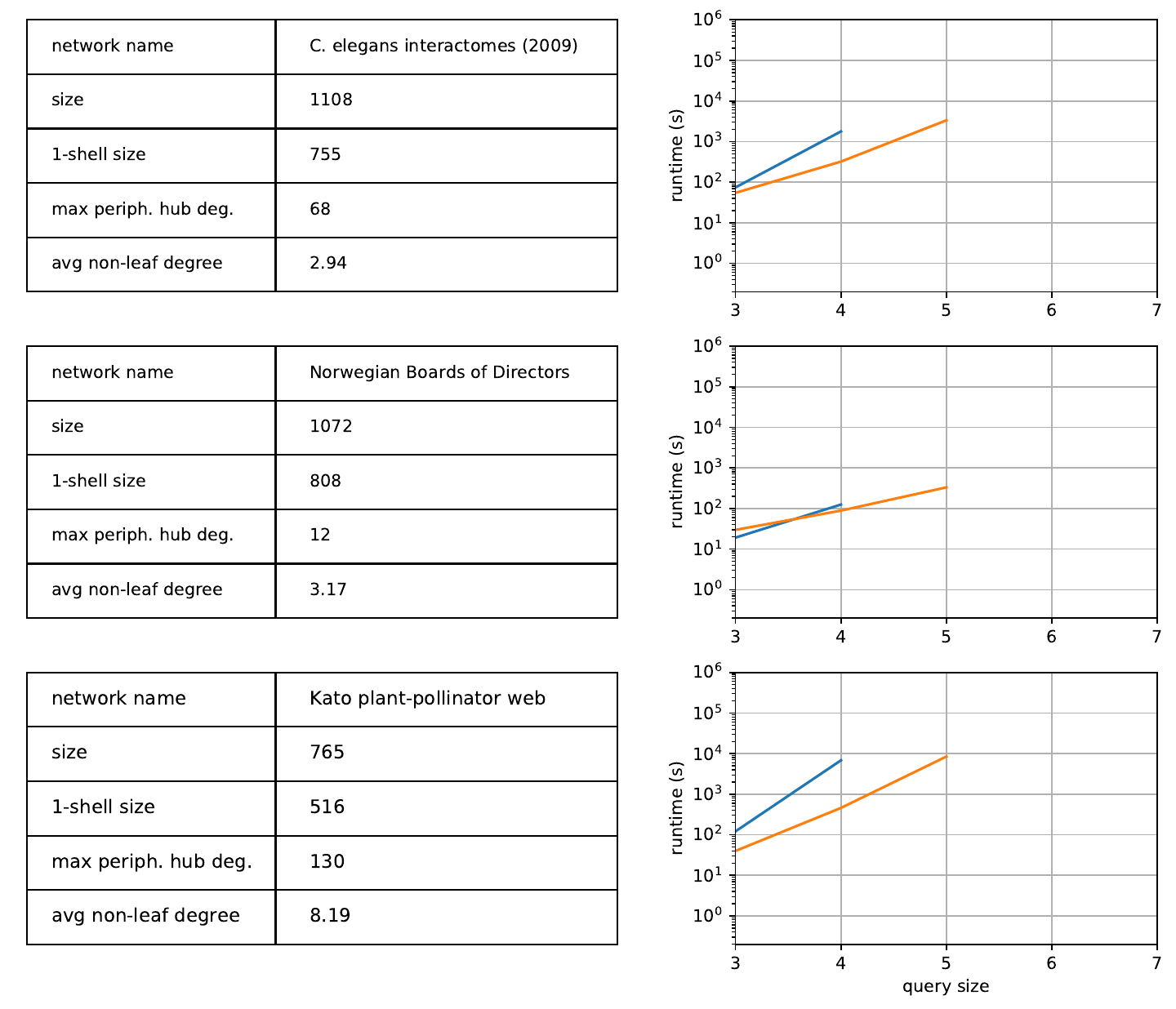}
	\caption{Similar to Figure~\ref{fig:runtime_comparison}, but for three larger graphs. Due to their size, we omit their pictures and onion decompositions. (Our Python implementation of the Grochow--Kellis algorithm did not finish on size-5 subgraphs within the bounds of our patience.)}
\end{figure}

\section{Conclusion and Future Work}
\label{sec:conclusion}
We designed and implemented an algorithm that counts instances of a given graph in a network, leverages tree and core-periphery structure, and can be adapted to any popular motif-centric algorithm. Our study suggests that network and query topology can have a great impact on the runtime of motif search algorithms, and alleviates this problem in the case of trees. However, hubs are still major contributors to the complexity of motif search, whether they appear at the peripheries of a network or within the core. A more detailed analysis of the relationship between distribution of runtimes and topologies of query graphs could reveal the best angles of attack for more instance-aware algorithms, which we leave for future work.

One natural way to extend our work is to use properties of $k$-shells for values of $k$ larger than one. Analogously to the decomposition of the 1-shell into trees, one can define an ``ear decomposition'' that maps the 2-shell onto a collection of paths that have only two vertices in common with the rest of the graph. This is equivalent to 2-edge-connectivity, and isomorphism of 2-edge-connected graphs is \textsc{GraphIsomorphism}-complete. However, there may be other properties of the 2-shell that enable practical speedup in spite of high worst-case complexity.

As previously discussed, our subtree-counting technique can be combined with any state-of-the-art motif-centric algorithm, with a small amount of adaptation. The main idea is to separate the 1-shell from the rest of a network, count subtrees in the 1-shell and run the chosen algorithm on the 2-core. To count subtrees that appear across the 1-shell/2-core boundary, some care has to be taken in modifying the algorithm to stop its search midway when it encounters a 1-shell root, call the subtree-counting method, and store the result for processing at the end of its search.

\bibliographystyle{siamplain}
\bibliography{references}

\end{document}